\documentclass[letterpaper, 10 pt, journal, twoside]{IEEEtran}
\IEEEoverridecommandlockouts

    \usepackage{cite}
	\usepackage{graphicx}
\usepackage{amsmath}
\usepackage{cite}
\usepackage{algorithm}
\usepackage{algpseudocode}

	\usepackage{amssymb}
	\usepackage{color}
	\usepackage[mathscr]{euscript}
	\usepackage{fixmath}
	\usepackage{url}
	\usepackage{tkz-graph}
	\usepackage{mdframed}
		 \usepackage{siunitx}
	 \usepackage{relsize}
	 \usepackage{subcaption}
	 \usepackage{bm}
	 \usepackage{siunitx}
	 \usepackage{wrapfig}
	 \usepackage{dsfont}
	 \usepackage{enumerate}
	 \usepackage{balance}
	 \usepackage{empheq}
	 \usepackage{mathtools}

\IEEEoverridecommandlockouts     

	\newcommand{\norm}[1]{\left\lVert#1\right\rVert}

\newtheorem{theorem}{Theorem}
\newtheorem{lemma}{Lemma}

\newtheorem{proof}{Proof}

\newtheorem{remark}{Remark}

\DeclareMathOperator*{\minimize}{minimize}
\DeclareMathOperator*{\st}{subject~to}

\renewcommand{\arraystretch}{1} 
	
	\DeclareMathAlphabet\mathbfcal{OMS}{cmsy}{b}{n}

	\title{An Input-Output Parametrization of Stabilizing Controllers: amidst Youla and System Level Synthesis}

	 \author{Luca Furieri$^\ast$\thanks{This research was gratefully funded by the European Union ERC Starting Grant CONENE. Antonis Papachristodoulou  
was supported in part by the EPSRC project EP/M002454/1. $^\ast$Luca Furieri and Maryam Kamgarpour are with the Automatic Control Laboratory, Department of Information Technology and Electrical Engineering, ETH Z\"{u}rich, Switzerland. E-mails: {\tt\footnotesize \{furieril, mkamgar\}@control.ee.ethz.ch}} \and Yang Zheng$^\dagger$\thanks{$^\dagger$Yang Zheng and Antonis Papachristodoulou are with the Department of Engineering Science , University of Oxford, United Kingdom. E-mails: {\tt\footnotesize \{yang.zheng, antonis\}@eng.ox.ac.uk}} \and Antonis Papachristodoulou$^\dagger$ \and Maryam Kamgarpour$^\ast$
}

\begin{document}

	\maketitle
	\thispagestyle{empty}
	
	\IEEEpeerreviewmaketitle

	\begin{abstract}
This paper proposes a novel input-output parametrization of the set of internally stabilizing output-feedback controllers for linear time invariant (LTI) systems. Our underlying idea is to directly treat the closed-loop transfer matrices from disturbances to input and output signals as design parameters and exploit their affine relationships.   This input-output perspective is particularly effective when a doubly-coprime factorization is difficult to compute, or an initial stabilizing controller is challenging to find; most previous work requires one of these pre-computation steps. Instead, our approach can bypass such pre-computations, in the sense that a stabilizing controller is computed by directly solving a linear program (LP).  Furthermore, we show that the proposed input-output parametrization allows for computing norm-optimal controllers subject to quadratically invariant (QI) constraints using convex programming.
	\end{abstract}	
	\section{Introduction}
	\label{sec:intro}
	\IEEEPARstart{G}iven a multi-input multi-output linear time invariant (MIMO LTI) system, a classical problem in control theory is to design an output-feedback controller that stabilizes the closed-loop system to external perturbations in the most efficient way. Solving the corresponding optimization problem is known to be computationally hard, partly due to the inherent non-convexity of the set of stabilizing controllers and partly due to the challenge of including additional constraints on the controller in a convex way  \cite{Survey,Witsenhausen}. 
	
 The renowned work \cite{rotkowitz2006characterization} established that, when the constraints on the output-feedback controller are quadratically invariant (QI) with respect to the system, one can compute norm-optimal stabilizing controllers using convex programming.   A limitation of the controller design procedure of \cite{rotkowitz2006characterization} is that the system is required to be \emph{strongly stabilizable}, i.e., a stabilizing output-feedback controller that itself is stable must exist and  be known in advance. Only then,  can a convex optimization problem be cast. However, it can be challenging to find such a stable and stabilizing controller. This gap in the controller design procedure was  addressed in~\cite{sabuau2014youla}, where the authors proposed a Youla-like parametrization~\cite{youla1976modern} to overcome the strong stabilizability assumption. We note that a doubly-coprime factorization  of the plant  must be computed as a preliminary step in~\cite{sabuau2014youla}. For a finite-dimensional and rational plant one can use the procedure of \cite{nett1984connection} to compute a doubly-coprime factorization. However, it was shown that a general internally stabilizable plant does
not necessarily admit a doubly-coprime factorization \cite{anantharam1985stabilization}, and computing one can be challenging even when one does exist \cite{laakkonen2016robust,foias1996robust}. Recently, \cite{sabau2017convex} recognized these difficulties and proposed adapting the so-called \emph{coordinate-free} approach \cite{mori2004elementary} for convex computation of controllers subject to strongly quadratically invariant (SQI) constraints. However, the SQI condition is more restrictive than QI, and the approach of \cite{sabau2017convex} requires an initial  stabilizing controller in advance, which may be challenging to find.
  
  Unlike~\cite{youla1976modern,sabuau2014youla,rotkowitz2006characterization,mori2004elementary,sabau2017convex}, which adopted a purely input-output perspective, \cite{wang2019system} proposed a detailed state-space point of view using the so-called system level approach to controller synthesis. An advantage of the state-space parametrization \cite{wang2019system} is that it does not depend on a doubly-coprime factorization of the system or the knowledge of an initial stabilizing controller. 
  Inspired by \cite{wang2019system}, we raise  the question of whether one could avoid  a detailed state-space perspective and instead adopt a purely input-output one in frequency domain, in order to eliminate the potentially challenging pre-computation steps in \cite{youla1976modern,sabuau2014youla,rotkowitz2006characterization,mori2004elementary,sabau2017convex}. In this paper, we present a positive answer to this question.
  
The contributions of this paper are as follows. First, we show that the set of all internally stabilizing controllers for a given LTI system can be  expressed as an affine subspace of four input-output parameters. Unlike the methods in \cite{youla1976modern,sabuau2014youla,rotkowitz2006characterization,mori2004elementary,sabau2017convex}, our input-output parametrization does not depend on a doubly-coprime factorization of the system or an initial stabilizing controller. Second, we  prove the  equivalence between our input-output parametrization and the classical Youla parametrization \cite{youla1976modern}. In particular, we derive the relationships between the proposed parameters and the Youla one in terms of any given doubly-coprime factorization,  highlighting the reason why computing a doubly-coprime factorization can be bypassed.  Third, we show that subspace constraints that are QI can be expressed in a convex way within the proposed input-output parametrization.  Last, we apply the proposed parametrization to the problem of computing $\mathcal{H}_2$ norm-optimal distributed controllers in the discrete-time and continuous-time domains.

The paper is structured as follows. Section~\ref{se:problem_statement} introduces the problem statement, and Section~\ref{se:Main_Results} presents our main theoretical findings. Numerical examples are used to illustrate our approach in Section~\ref{se:implementation}. We conclude the paper in Section~\ref{se:conclusion}.

\emph{Notation:}      		We denote the imaginary axis as	
    $
    j\mathbb{R}:=\{z \in \mathbb{C} \mid \Re(z)=0\}\,,
    $
     and  the unit circle as 
     $e^{j\mathbb{R}}:=\{z \in \mathbb{C}\mid \Re(z)^2+\Im(z)^2=1\}\,.$	
	We consider continuous-time and discrete-time transfer functions, defined as rational functions $g_c:j\mathbb{R} \rightarrow \mathbb{C}$ and $g_d:e^{j\mathbb{R}} \rightarrow \mathbb{C}$ respectively. A transfer function is called {\emph{proper} (resp. \emph{strictly-proper})} if the degree of the numerator polynomial does not exceed {(resp. is strictly lower than)} the degree of the denominator polynomial. Upon denoting $s=j\omega$ and $z=e^{j \omega}$, we define the \emph{poles} of $g_c$ and $g_d$ as the roots of the denominator polynomials of $g_c$ and $g_d$.  Similar to~\cite{rotkowitz2006characterization}, we denote by $\mathcal{R}_p^{m \times n}$ the set of $m \times n$ proper \emph{transfer matrices}, that is the set of $m \times n$ matrices whose entries are  proper transfer functions. Also, we denote by $\mathcal{R}_{sp}^{m \times n}$ the set of $m \times n$ strictly proper transfer matrices. Finally, we let $\mathcal{RH}_\infty^{m \times n}$ be the set of $m \times n$ proper \emph{stable} transfer matrices. For continuous-time systems we have
	\begin{equation*}
	\mathcal{R}\mathcal{H}_\infty^{m \times n}:=\{G \in \mathcal{R}_p^{m\times n}|~G\text{ has no poles in $\mathbb{C}_+$} \}\,,
	\end{equation*}
	while for discrete-time systems we have
		\begin{equation*}
	\mathcal{R}\mathcal{H}_\infty^{m \times n}:=\{G \in \mathcal{R}_p^{m\times n}|~G\text{ has no poles in $\mathbb{C}_{||\cdot||_2\geq 1}\}$ }\,,
	\end{equation*}
	where $\mathbb{C}_+=\{z \in \mathbb{C}|~\Re(z)\geq 0\}$ and $\mathbb{C}_{||\cdot||_2\geq 1}=\{z \in \mathbb{C}|~\Re(z)^2+\Im(z)^2 \geq 1\}$. 
	\section{Problem Statement}
	\label{se:problem_statement}
	We consider the feedback system with a block-structure shown in Figure~\ref{fig:feedback}, where $G$ and $K$ represent a system and a feedback controller, respectively. 
		\begin{figure}[t]
	      \centering
	      	\includegraphics[width=0.4\textwidth]{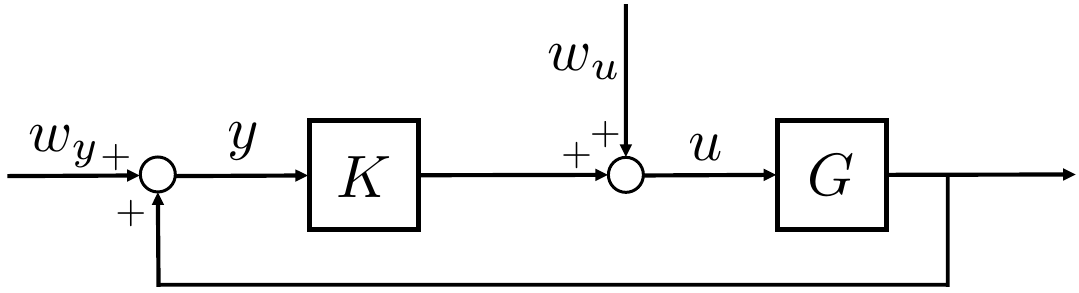}
	      	\caption{ \small Standard feedback interconnection of system and controller.}
	      		\label{fig:feedback}
	\end{figure}
	For causality and well-posedness of the problem, we assume that the system $G \in \mathcal{R}_{sp}^{p \times m}$ is strictly proper and that $K \in \mathcal{R}_{p}^{m \times p}$ is proper. Such assumptions guarantee that the inverses of $(I-GK)$ and $(I-KG)$ exist \cite{rotkowitz2006characterization}, thus ensuring well-posedness and causality.

	The linear system in Figure~\ref{fig:feedback} is equivalently described by the following equations
\begin{equation}
\begin{cases}
y=Gu+w_y \,,\\
u=Ky+w_u\,, \\
\end{cases} \label{eq:LTI}
\end{equation}			
	where $y$ is the output signal vector of dimension $p\in \mathbb{N}$, $u$ is the input signal vector of dimension $m \in \mathbb{N}$ and $w_y,w_u$ represent external disturbances of dimensions $p$ and $m$ respectively. By reorganizing (\ref{eq:LTI}) we obtain the closed-loop equations
	\begin{equation}
		\label{eq:closed_loop}
\begin{bmatrix}
y\\u
\end{bmatrix}=\begin{bmatrix}
(I-GK)^{-1}&(I-GK)^{-1}G\\
K(I-GK)^{-1}&(I-KG)^{-1}
\end{bmatrix}\begin{bmatrix}
w_y\\w_u
\end{bmatrix}\,.
	\end{equation}
	 The controller $K$ is said to be \emph{internally stabilizing} for the system $G$ if and only if the four transfer matrices in (\ref{eq:closed_loop}) are all stable \cite{francis1987course}. We thus define the set of internally  stabilizing controllers as follows
	\begin{equation*}
	\mathcal{C}_{\text{stab}}=\{K \in \mathcal{R}_p^{m \times p}|~K \text{ internally stabilizes } G\}\,.
	\end{equation*}
It is known that the set $\mathcal{C}_{\text{stab}}$ is non-convex in general. This can be easily verified, for instance, by selecting $K_1,K_2 \in \mathcal{C}_{\text{stab}}$ and noticing that $\frac{1}{2}(K_1+K_2)$ does not internally stabilize $G$ in general. Hence, directly computing a controller $K$ in $\mathcal{C}_{\text{stab}}$ using convex programming is not possible. The best-known method to obtain a convex parametrization of $\mathcal{C}_{\text{stab}}$ is the  Youla parametrization \cite{youla1976modern}, which relies on pre-computing a doubly-coprime factorization of $G$. The main goal of this paper is to present a new input-output parametrization for $\mathcal{C}_{\text{stab}}$ that is defined and implemented directly without  pre-computation steps.
	\section{Affine Parametrization of Stabilizing Controllers}
	\label{se:Main_Results}
	In this section we present our main results. First, we show that $\mathcal{C}_{\text{stab}}$ can be expressed as an affine subspace, without the need of computing a doubly-coprime parametrization of the system \cite{youla1976modern,sabuau2014youla} or a stabilizing controller \cite{sabau2017convex} in advance. Second, we derive an explicit connection with the classical Youla-parametrization and establish the reason why a doubly-coprime factorization is not necessary. Last, we show that our parametrization recovers the results of \cite{rotkowitz2006characterization,sabuau2014youla} on including  subspace constraints that are QI  in an exact and convex way.
	\subsection{An input-output convex parametrization of internally stabilizing controllers}
	The closed-loop equations (\ref{eq:closed_loop}) are equivalent to
\begin{equation}
	\label{eq:closed_loop_parameters}
\begin{bmatrix}
y\\u
\end{bmatrix}=\begin{bmatrix}
X&W\\
Y&Z
\end{bmatrix}\begin{bmatrix}
w_y\\w_u
\end{bmatrix}\,,
	\end{equation}
	where $(X,Y,W,Z$) are all functions of the system $G$ and the controller $K$ as per (\ref{eq:closed_loop}). 
	Our main idea is to treat the closed-loop transfer matrices $(X,Y,W,Z)$ in (\ref{eq:closed_loop_parameters}) directly as design parameters, and to  exploit their mutual relationships in terms of $G$ and $K$. We thus present our first theorem, whose proof is reported in the Appendix.
	\begin{theorem}
	\label{th:parametrization}
	Consider the LTI (\ref{eq:LTI}). The following statements hold.
	\begin{enumerate}
		 \item For any $K \in \mathcal{C}_{\text{stab}}$ there exist four corresponding transfer matrices $(X,Y,W,Z)$ that lie in the affine subspace defined by the equations
	 \begin{subequations}
	 \begin{align}
	 &\begin{bmatrix}
	 I&-G
	 \end{bmatrix}\begin{bmatrix}
	 X&W\\Y&Z
	 \end{bmatrix}=\begin{bmatrix}
	 I&0
	 \end{bmatrix}\,, \label{eq:aff1}\\
	 & \begin{bmatrix}
	 X&W\\Y&Z
	 \end{bmatrix}\begin{bmatrix}
	 -G\\I
	 \end{bmatrix}=\begin{bmatrix}
	 0\\I
	 \end{bmatrix}\label{eq:aff2}\,,\\
	 &\begin{matrix}
	 X\in \mathcal{RH}_\infty^{p \times p}\,,&Y\in \mathcal{RH}_\infty^{m \times p}\,,\\
	 ~W \in \mathcal{RH}_\infty^{p \times m}\,,&Z \in \mathcal{RH}_\infty^{m \times m}\,.
	 \end{matrix}\label{eq:aff3}
	 \end{align}
	 \end{subequations}
	 	 \item For any transfer matrices $(X,Y,W,Z)$ that lie in the affine subspace (\ref{eq:aff1})-(\ref{eq:aff3}), the controller $K=YX^{-1}$ belongs to $\mathcal{C}_{\text{stab}}$.
	\end{enumerate}
	\end{theorem}
	
	Theorem~\ref{th:parametrization} leads to a novel input-output parametrization of all internally stabilizing controllers as an affine subspace.  
	
	
	\begin{remark}[Optimal controller synthesis]
 A common scenario (e.g. \cite{rotkowitz2006characterization,sabuau2014youla},\cite{ sabau2017convex,wang2019system})  involves  a disturbance $w$ of dimension $r \in \mathbb{N}$ such that $w_y=P_{yw}w$ for a transfer matrix $P_{yw} \in \mathcal{R}_p^{p \times r}$  and a  performance signal $z$ of dimension $q \in \mathbb{N}$ such that $z=P_{zu}u+P_{zw}w$ for  transfer matrices $P_{zu}\in \mathcal{R}_p^{q \times m}$, $P_{zw} \in \mathcal{R}_p^{q \times r}$. The goal is to minimize a given norm of the closed-loop transfer function from $w$ to $z$. This quantity can be encoded in terms of the parameter $Y$ as   $
\|P_{zw}+P_{zu}YP_{yw}\|\,,
$
 where $\|\cdot\|$ is any norm of interest. 
 Then, by Theorem~\ref{th:parametrization}, the optimal stabilizing controller in $\mathcal{C}_{\text{stab}}$  is found by solving the  convex program
	\begin{alignat}{3}
	   & \minimize_{X,Y,W,Z } && \|P_{zw}+P_{zu}YP_{yw}\| \nonumber\\
	   & \st&&~~~(\ref{eq:aff1})-(\ref{eq:aff3}) \nonumber\,.
	\end{alignat}	

	\end{remark}
	
	\subsection{Equivalence with  Youla: beyond doubly-coprime factorizations of the system}
	As discussed in Section~\ref{sec:intro}, other parametrizations of internally stabilizing controllers require preliminary knowledge of either a strongly stabilizing controller \cite{rotkowitz2006characterization}, a doubly-coprime factorization of the system \cite{youla1976modern,sabuau2014youla} or a stabilizing controller \cite{mori2004elementary,sabau2017convex}. Instead, Theorem~\ref{th:parametrization} establishes that $\mathcal{C}_{\text{stab}}$ can be parametrized as an affine subspace (\ref{eq:aff1})-(\ref{eq:aff3}) that depends on the transfer matrix $G$ directly. This result might surprise the reader familiar with the Youla parametrization and the  previous approaches \cite{youla1976modern,sabuau2014youla,mori2004elementary,sabau2017convex}. Here, we shed light on this desirable feature of our input-output  parametrization. First, we recall the notion of a doubly-coprime factorization of the system from \cite[Chapter 4]{francis1987course}:

\begin{lemma}
\label{le:doubly}
For any  $G \in \mathcal{R}_{sp}^{p \times m}$ there exist eight proper and stable transfer matrices satisfying the equations
\begin{align}
&G=N_rM_r^{-1}=M_l^{-1}N_l\,,\label{eq:dp1}\\
&\begin{bmatrix}
U_l&-V_l\\
-N_l&M_l
\end{bmatrix}\begin{bmatrix}
M_r&V_r\\
N_r&U_r
\end{bmatrix}=I_{m+p}\,.\label{eq:dp2}
\end{align}
\end{lemma}
Then, the Youla parametrization  of all internally stabilizing controllers \cite{youla1976modern} establishes the following equivalence:
\begin{equation*}
\mathcal{C}_{\text{stab}}=\{(V_r-M_rQ)(U_r-N_rQ)^{-1}|~Q \in \mathcal{RH}_\infty^{m \times p}\}\footnote{Equivalently, $\mathcal{C}_{\text{stab}}=\{(U_l-QN_l)^{-1}(V_l-QM_l)|~Q \in \mathcal{RH}_\infty^{m \times p}\}$.}\,.
\end{equation*}
In other words, for a given doubly-coprime factorization (\ref{eq:dp1})-(\ref{eq:dp2}), the non-convex set $\mathcal{C}_{\text{stab}}$ can be expressed in terms of the linear subspace of stable Youla transfer matrices $Q \in \mathcal{RH}_\infty^{m \times p}$. Next, we present our result on equivalence between the Youla parametrization and the input-output parametrization. Its proof is reported in the Appendix.
\begin{theorem}
\label{th:Youla_eq}
Let $U_r,V_r,U_l,V_l,M_r,M_l,N_r,N_l$ be any  doubly-coprime factorization of the system $G \in \mathcal{R}_{sp}^{p \times m}$. The following statements hold.
\begin{enumerate}
\item For any $Q \in \mathcal{RH}_\infty^{m \times p}$, the following transfer matrices
\begin{subequations}
\begin{align}
&X=(U_r-N_rQ)M_l\,, \label{eq:Q_to_X_1}\\
&Y=(V_r-M_rQ)M_l\,,\\
&W=(U_r-N_rQ)N_l\,,\\
&Z=I+(V_r-M_rQ)N_l\,,\label{eq:Q_to_X_2}
\end{align}
\end{subequations}
belong to the affine subspace (\ref{eq:aff1})-(\ref{eq:aff3}) and are such that $YX^{-1}=(V_r-M_rQ)(U_r-N_rQ)^{-1}$.
\item For any $(X,Y,W,Z)$ in the affine subspace  (\ref{eq:aff1})-(\ref{eq:aff3}), the transfer matrix
\begin{equation}
\label{eq:Youla_with_XYWZ}
Q=V_lXU_r-U_lYU_r-V_lWV_r+U_lZV_r-V_lU_r\,,
\end{equation}
is such that $Q \in \mathcal{RH}_\infty^{m \times p}$ and $(V_r-M_rQ)(U_r-N_rQ)^{-1}=YX^{-1}$.
\end{enumerate}
\end{theorem}

Theorem~\ref{th:Youla_eq} establishes the equivalence of the proposed parametrization with the Youla parametrization for any existing doubly-coprime factorization of $G$. The relationships (\ref{eq:Q_to_X_1})-(\ref{eq:Q_to_X_2}) and (\ref{eq:Youla_with_XYWZ}) offer an important insight: they encapsulate that stability of $(X,Y,W,Z)$ in the affine subspace (\ref{eq:aff1})-(\ref{eq:aff3}) is equivalent to stability of the Youla parameter corresponding to \emph{any} doubly-coprime factorization of the system. 

Also notice that (\ref{eq:aff1})-(\ref{eq:aff3}) are equivalent to considering one parameter $Y$ and enforcing  that $Y$, $I+GY$, $I+YG$ and $G+GYG$ are all stable transfer matrices. Then, in accordance with Theorem~\ref{th:parametrization}, our method is to parametrize $\mathcal{C}_{\text{stab}}$ by explicitly imposing stability of four affine functions of  $Y$ corresponding to the four closed-loop transfer matrices. On the other hand, the Youla parametrization requires stability of only one $m \times p$ parameter by exploiting a pre-computed doubly-coprime factorization that maps any stable $Q\in \mathcal{RH}_\infty^{m \times p}$ back to $\mathcal{C}_{\text{stab}}$. In conclusion, we have shown that the mathematical construct of doubly-coprime factorizations can be bypassed by instead directly adding affine constraints in the form (\ref{eq:aff1})-(\ref{eq:aff3}) within our suggested parametrization. 

Finally, we remark that in this paper we assumed that $G$ is rational and finite-dimensional for simplicity. This  ensures that a doubly-coprime factorization of $G$ exists by Lemma~\ref{le:doubly} and can be found with standard methods \cite{nett1984connection}. For a general plant $G$,  a doubly-coprime factorization might not exist \cite{anantharam1985stabilization} or be challenging to find \cite{laakkonen2016robust,foias1996robust}. Instead, our input-output parametrization  remains defined directly in terms of $G$. 

	\subsection{Quadratically invariant constraints on the stabilizing controllers}
	It was shown that Youla-like parametrizations \cite{sabuau2014youla,rotkowitz2006characterization}  can be used to encode, in a convex way, subspace constraints on $K$ that are \emph{quadratically invariant} (QI) with respect to the system. The coordinate free approach \cite{mori2004elementary} instead allows for a slightly less general result, in that \emph{strongly quadratically invariant} (SQI) constraints \cite{sabau2017convex} can be encoded in a convex way.  Here, we show that our input-output parametrization  allows for a straightforward inclusion of subspace constraints that are QI in a convex way. 
	
	We begin by reviewing the notion of QI. A set $\mathcal{K} \subseteq \mathcal{R}_p^{m \times p}$ is said to be \emph{Quadratically Invariant (QI)} with respect to the system $G$ if and only if
    \begin{equation*}
    KGK \in \mathcal{K}, \quad \forall K \in \mathcal{K}\,. 
    \end{equation*} 
 Motivated by \cite{rotkowitz2006characterization,sabuau2014youla}, we define the closed-loop transformation $h_G:\mathcal{R}_p^{m \times p} \rightarrow \mathcal{R}_p^{m \times p}$ of $G$ with $K$ :
\begin{equation*}
h_G(K):=-K(I-GK)^{-1}\,.
\end{equation*}
  We then recall the following main result from \cite[Theorem~14]{rotkowitz2006characterization} and \cite[Theorem~9]{lessard2010algebraic}.
	\begin{lemma}
	\label{le:QI}
	Let $\mathcal{K} \subseteq \mathcal{R}_p^{m \times p}$ be a subspace. Then
\begin{equation*}
\mathcal{K}\text{ is QI w.r.t. }G \iff h_G(\mathcal{K})=\mathcal{K}\,.
\end{equation*}	
	\end{lemma}
	
We are now ready to present our  result about including subspace constraints that are QI with respect to $G$. 
\begin{theorem}
\label{th:QI}
Let $\mathcal{K} \subseteq \mathcal{R}_p^{m \times p}$ be a subspace that is QI  with respect to $G$ . Then, any controller  $K \in \mathcal{C}_{\text{stab}} \cap \mathcal{K}$ is represented by the affine subspace
\begin{equation}
\{(X,Y,W,Z)|~(\ref{eq:aff1})-(\ref{eq:aff3}),~Y \in \mathcal{K}\}\,,\label{eq:affine_subspace_sparse}
\end{equation}
and every $(X,Y,W,Z)$ in (\ref{eq:affine_subspace_sparse}) corresponds to a controller $K \in \mathcal{C}_{\text{stab}} \cap \mathcal{K}$.
\end{theorem}
\begin{proof}
Let $K \in \mathcal{C}_{\text{stab}} \cap \mathcal{K}$ and choose $X=(I-GK)^{-1}$, $Y=K(I-GK)^{-1}$, $W=(I-GK)^{-1}G$ and $Z=(I-KG)^{-1}$. By Theorem~\ref{th:parametrization}, we have that $(X,Y,W,Z)$ satisfies (\ref{eq:aff1})-(\ref{eq:aff3}). Notice that $Y=-h_G(K)$ by definition. By Lemma~\ref{le:QI}, since $K \in\mathcal{K}$ and $\mathcal{K}$ is a  subspace that is QI with respect to $G$, we have that $h_G(K) \in \mathcal{K}$. Hence, $Y \in \mathcal{K}$ and $(X,Y,W,Z)$ belongs to the set (\ref{eq:affine_subspace_sparse}).

 Vice-versa, let $(X,Y,W,Z)$ lie in the set (\ref{eq:affine_subspace_sparse}) and choose $K=YX^{-1}$. By Theorem~\ref{th:parametrization}, we have that $K \in \mathcal{C}_{\text{stab}}$. Notice that by using (\ref{eq:aff1}) we have $K=Y(I+GY)^{-1}=h_G(-Y)$.  By Lemma~\ref{le:QI}, since $-Y \in \mathcal{K}$ and $\mathcal{K}$ is a subspace that is QI with respect to $G$, we have that $h_G(-Y) \in \mathcal{K}$. Hence, $K \in \mathcal{K}$ and  $K \in \mathcal{C}_{\text{stab}} \cap \mathcal{K}$ as desired.
\end{proof}


Theorem~\ref{th:QI} shows that if $\mathcal{K}$ is a subspace that is QI with respect to $G$, simply adding the requirement $Y \in \mathcal{K}$ to the constraints (\ref{eq:aff1})-(\ref{eq:aff3}) allows for parametrizing all the internally stabilizing controllers $K \in \mathcal{C}_{\text{stab}} \cap \mathcal{K}$.

\section{Implementing The Input-Output Parametrization}
\label{se:implementation}
Here, we first discuss how controllers in the affine subspace (\ref{eq:aff1})-(\ref{eq:aff3}) can be obtained in practice by solving a linear program (LP). Then, we investigate application of the input-output parametrization to efficient computation of norm-optimal distributed controllers.  
\subsection{Linear programming for stabilizing controllers}
\label{sub:LP}
Despite being affine, the subspace (\ref{eq:aff1})-(\ref{eq:aff3}) is infinite-dimensional because the order of the polynomials in the entries of $(X,Y,W,Z)$ is not fixed. We use the  results of \cite{linnemann1999convergent,pohl2009advanced} to obtain finite-dimensional approximations in continuous- and discrete-time that can be solved efficiently.

Consider the sets  
\begin{align}
&\left\{(s+a)^{-k}\right\}_{k=0}^N\,,\label{eq:infinite_basis_continuous}\\
&\left\{z^{-k}\right\}_{k=0}^N\,,\label{eq:infinite_basis_discrete}
\end{align}
 where $N \in \mathbb{N}$ and $a>0$ is any real number. By \cite{linnemann1999convergent}, for every $g' \in \mathcal{RH}_\infty$  in continuous-time there exists $g$ in the subspace spanned by (\ref{eq:infinite_basis_continuous}) with $N \rightarrow \infty$ such that $\|g \hspace{-0.1cm}- \hspace{-0.1cm}g'\| \hspace{-0.1cm}< \hspace{-0.1cm}\epsilon$ for every $\epsilon>0$, where $\| \cdot  \|$ can be, for instance, the $\mathcal{H}_2$, $\mathcal{H}_\infty$ or $\mathcal{L}_1$ norm. Hence, optimizing over the subspace spanned by  (\ref{eq:infinite_basis_continuous}) for $N \rightarrow \infty$ yields the same results as optimizing over $\mathcal{RH}_\infty$. For the discrete-time case, it is known that (\ref{eq:infinite_basis_discrete}) spans the entire $\mathcal{R}_p$ for $N\rightarrow \infty$ \cite[Theorem 4.7]{pohl2009advanced}. According to these results, a  finite-dimensional approximation of the $(X,Y,W,Z)$'s satisfying (\ref{eq:aff3}) is obtained by requiring that  they are expressed as
 \begin{equation}
 \label{eq:FIR_constraint}
\qquad \quad \begin{matrix}X=\sum_{i=0}^NX[i]\sigma^{-i}\,, &Y=\sum_{i=0}^NY[i]\sigma^{-i}\,, \\
W=\sum_{i=0}^NW[i]\sigma^{-i}\,,& Z=\sum_{i=0}^NZ[i]\sigma^{-i}\,, \end{matrix}
\end{equation}
 where the real matrices $(X[i]$, $Y[i]$, $W[i]$, $Z[i])$ for all $i=0,\ldots,N$ are the decision variables,  and we pose $\sigma\hspace{-0.05cm}=\hspace{-0.05cm}(s+a)$ in continuous-time and $\sigma=z$ in discrete-time.  
The additional relationships (\ref{eq:aff1})-(\ref{eq:aff2}) result in a set of affine constraints on $(X[i],Y[i],W[i],Z[i])$ for all $i=0,\ldots,N$, which are obtained by setting the coefficients of the numerator polynomials appearing in each entry of the transfer matrices $X -I-GY$, $W-GZ$, $-XG+W$ and $-YG+Z-I$ to zero. We note that, in accordance with the results of \cite{wang2019system}, such affine constraints are feasible for values of $N$ that are large enough; pre-computing feasible values of $N$ and a formal analysis of the convergence of the minimum cost  to that of the full original problem as $N \hspace{-0.05cm}\rightarrow \hspace{-0.05cm}\infty$ are beyond the scope of this paper, but a relevant direction for future research. We conclude that transfer matrices in the affine subspace (\ref{eq:aff1})-(\ref{eq:aff3}) can be computed efficiently by implementing a corresponding LP based on the finite-dimensional assumption (\ref{eq:FIR_constraint}).

  \subsection{Application examples: norm-optimal distributed control}
  Let us  consider the following unstable system in discrete-time taken from \cite{Alavian}:
\begin{equation}
\label{eq:system}
G_d=
\begingroup
\renewcommand*{\arraystretch}{0.7}
\begin{bmatrix}
v(z)&0&0&0&0\\
v(z)&u(z)&0&0&0\\
v(z)&u(z)&v(z)&0&0\\
v(z)&u(z)&v(z)&v(z)&0\\
v(z)&u(z)&v(z)&v(z)&u(z)
\end{bmatrix}\,,
\endgroup
\end{equation}
where $v(z)=\frac{0.1}{z-0.5}$ and $u(z)=\frac{1}{z-2}$. Let us also consider a continuous-time unstable system $G_c$, defined in the same way as per (\ref{eq:system}) with the substitutions $u(s)=\frac{1}{s-1}$ and $v(s)=\frac{1}{s+1}$ instead of $u(z)$ and $v(z)$, respectively (taken from \cite{rotkowitz2006characterization}). Our goal is to compute a distributed stabilizing controller $K$ that minimizes a cost function to be defined and complies with a desired \emph{sparsity pattern}, that is, some specific entries of  the transfer matrix $K$ must be 0 to encode the fact that certain scalar control inputs cannot measure certain outputs. Formally, we require that $K \in \mathcal{C}_{\text{stab}} \cap \text{Sparse}(S)$, where $S\in \{0,1\}^{m \times p}$ is a given binary matrix and 
\begin{equation*}
\text{Sparse}(S)=\{K \in \mathcal{R}_p^{m \times p}|~K_{ij}=0 \text{ if } S_{ij}=0\}\,,
\end{equation*}
is a subspace. Here, we consider the sparsity pattern
\begin{equation*}
S=
\begingroup
\renewcommand*{\arraystretch}{0.8}
\begin{bmatrix}
1&0&0&0&0\\
1&1&0&0&0\\
1&1&1&0&0\\
1&1&1&1&0\\
1&1&1&1&1
\end{bmatrix}\,,
\endgroup
\end{equation*}
which is also considered in \cite{rotkowitz2006characterization,Alavian}. It is easy to verify that $\text{Sparse}(S)$ is QI with respect to $G_d$ and $G_c$  \cite[Theorem~26]{rotkowitz2006characterization}. 

\emph{Cost function:} We consider the cost function of \cite{rotkowitz2006characterization,Alavian}, chosen as $\|P_{zw}+P_{zu}K(I-GK)^{-1}P_{yw}\|_{\mathcal{H}_2}$, where
\begin{equation*}
P_{zw}=\begin{bmatrix}
G&0\\0&0
\end{bmatrix}\,, \quad P_{zu}=\begin{bmatrix}
G\\I
\end{bmatrix}\,, \quad P_{yw}=\begin{bmatrix}
G&I
\end{bmatrix}\,,
\end{equation*}
and ``$G$'' stands for either $G_c$ or $G_d$. The meaning of this cost function was explained in Remark~1, Section~\ref{se:Main_Results}. Using (\ref{eq:closed_loop}), (\ref{eq:closed_loop_parameters}) the cost function is equivalent to
\begin{equation}
\label{eq:cost_sensitivity}
\norm{ \begin{bmatrix}
W&X-I\\
Z-I&Y
\end{bmatrix}}_{\mathcal{H}_2}\,,
\end{equation}
which is convex in $(X,Y,W,Z)$. By Theorem~\ref{th:parametrization} and Theorem~\ref{th:QI}, the optimal control problem under investigation is reformulated as the following convex program:
	\begin{alignat}{3}
	   & \minimize_{X,Y,W,Z } &&~~~(\ref{eq:cost_sensitivity}) \label{probl:tosolve}\\
	   & \text{subject to} &&~~~(\ref{eq:aff1})-(\ref{eq:aff3}),~~Y \in \text{Sparse}(S) \nonumber\,.
	\end{alignat}	
	The program above is infinite-dimensional. Next, we exploit the finite-dimensional approximation (\ref{eq:FIR_constraint}) for both the discrete-time and continuous-time cases. All the resulting optimization problems were solved with MOSEK \cite{mosek}, called through MATLAB via YALMIP \cite{YALMIP}, on a computer with a 16GB RAM and a 4.2 GHz quad-core Intel i7 processor.

\subsubsection{Discrete-time case}

By (\ref{eq:FIR_constraint}) and  the definition of the $\mathcal{H}_2$ norm in discrete-time, (\ref{eq:cost_sensitivity}) can be written as
\begin{equation}
\sum_{i=0}^N\text{Trace}(J[i]^\mathsf{T}J[i])\,, \label{eq:trace}
\end{equation}
where
\begin{equation*}
J[0]=\begin{bmatrix}
W[0]&X[0]-I\\Z[0]-I&Y[0]
\end{bmatrix}\,,~J[i]=\begin{bmatrix}
W[i]&X[i]\\Z[i]&Y[i]
\end{bmatrix}\,,
\end{equation*}
for each $i=1,\ldots,N$. The cost function (\ref{eq:trace}) is thus quadratic in $(X[i],Y[i],W[i],Z[i])$ for every $i=1,\ldots,N$. The affine constraints (\ref{eq:aff1})-(\ref{eq:aff3}) are implemented as outlined in Section~\ref{sub:LP}, while $Y \in \text{Sparse}(S)$ is enforced by setting $Y[i]_{jk}=0$ for every $i=1,\ldots,N$ and $j,k$ such that $S_{jk}=0$. Problem (\ref{probl:tosolve}) is thus reduced to a quadratic program (QP), efficiently solvable with off-the-shelf software.

\emph{Simulation: }We set the order in (\ref{eq:FIR_constraint}) to $N=10$, as higher values for $N$ brought negligible improvement on the minimum cost. First, we omitted the sparsity constraints and obtained a centralized closed-loop $\mathcal{H}_2$ norm of $5.67$. Next, we computed the optimal distributed controller $K \in \text{Sparse}(S)$ and obtained a closed-loop $\mathcal{H}_2$ norm of $6.73$. In both cases, the solver time was less than $1$ second.


\subsubsection{Continuous-time case}

 Since the system $G_c$ is defined in continuous-time, the cost function does not admit the form (\ref{eq:trace}). Encoding the $\mathcal{H}_2$ norm in continuous-time presents an additional challenge. 
  As outlined in Section~\ref{sub:LP}, our LP based computation  offers a  solution to the main difficulty in implementing the coordinate-free approach of \cite{sabau2017convex}, that is obtaining an initial controller $K_0 \in \mathcal{C}_{\text{stab}} \cap \mathcal{K}$. Once $K_0 \in \mathcal{C}_{\text{stab}} \cap \mathcal{K}$ is obtained by solving an LP within the input-output parametrization,  the convex model-matching formulation of \cite[Theorem IV.12]{sabau2017convex} can be exploited directly to optimize over.

 \emph{Simulation: }  We implemented the LP given by (\ref{eq:aff1})-(\ref{eq:aff3}) with $Y \in \text{Sparse}(S)$ and the finite-dimensional assumption (\ref{eq:FIR_constraint}), as described in Section~\ref{sub:LP}. At this stage we are only interested in computing any feasible solution, so we chose the smallest value $N=2$ that yields a feasible program. We also chose $a=3$. The LP was solved in $0.24$ seconds of solver time. Its solution gives the controller
\begin{equation*}
K_0=
\begingroup
\renewcommand*{\arraystretch}{0.8}
\frac{8}{s+7}\begin{bmatrix}
0&0&0&0&0\\
0&-2&0&0&0\\
0&0&0&0&0\\
0&1&0&0&0\\
0&\frac{2(s+5)(s+3)}{(s+1)(s+7)}&0&0&-2
\end{bmatrix}\,,
\endgroup
\end{equation*}
which can be verified to lie in $\mathcal{C}_{\text{stab}} \cap \text{Sparse}(S)$. Having computed an initial stabilizing controller $K_0$, the convex model-matching problem \cite[Theorem IV.12]{sabau2017convex} can  be cast. Since $K_0$ is itself stable, we solved this convex program through the  numerical technique of \cite{Alavian}, which is based on semi-definite programming. First, we omitted the sparsity constraints and obtained a centralized closed-loop $\mathcal{H}_2$ norm of $6.38$. Next, we computed the optimal distributed controller $K \in \text{Sparse}(S)$ and obtained a closed-loop $\mathcal{H}_2$ norm of $7.36$. Both results match  those of \cite[Figure 4]{rotkowitz2006characterization}, where the same system and sparsity patterns were considered. The solver time did not exceed $7$ seconds. 

 To summarize, the input-output parametrization allows for an LP-based computation of stabilizing controllers subject to QI subspace constraints by directly using the expression of the plant $G$. Furthermore, in discrete-time the $\mathcal{H}_2$ norm minimization problem can be cast as a QP in our suggested parameters. On the other hand, previous results and tools, e.g. \cite{sabau2017convex,Alavian}, can be exploited to cast and solve the same problem in continuous-time; since these techniques typically require knowledge of an initial stabilizing controller, our parametrization offers an LP-based solution to fill this gap in the design process. 

\section{Conclusion}
\label{se:conclusion}
We proposed an input-output parametrization of all internally stabilizing controllers subject to subspace constraints that are QI. A main advantage of the proposed parametrization is that it allows for bypassing potentially challenging pre-computation steps that were needed in previous approaches. This fact was shown by establishing the direct relationships (\ref{eq:Q_to_X_1})-(\ref{eq:Q_to_X_2}) and (\ref{eq:Youla_with_XYWZ})  of each of our parameters with the Youla parameter. Finally, the input-output parametrization clarifies that an internal-state point of view as per \cite{wang2019system} is not necessary to bypass pre-computations; instead, an input-output perspective in the frequency domain is sufficient. We provided numerical examples to illustrate applicability of the proposed parametrization. This work opens up the question of whether other useful parametrizations which generalize both the one we proposed and the system level parametrization \cite{wang2019system} can be established, in order to incorporate different classes of constraints and objectives in the control system design.  Furthermore, it would be relevant to address application of the proposed input-output  parametrization to the design of localized controllers \cite{wang2019system}. Last, we remark that the suggested parametrization could naturally address the case of non-rational and infinite-dimensional plants, as the relationships (\ref{eq:aff1})-(\ref{eq:aff3}) hold regardless of the domain of $G$. For such cases, the standard state-space representation used in the system level parametrization \cite{wang2019system} is not available and a doubly-coprime factorization might not exist \cite{anantharam1985stabilization} or be difficult to find \cite{laakkonen2016robust,foias1996robust}. However, practical implementation of the constraints (\ref{eq:aff1})-(\ref{eq:aff3}) when the plant  is non-rational and/or infinite-dimensional requires further study.



\section*{Appendix}
\subsection{Proof of Theorem~\ref{th:parametrization}}

1) Let $K$ be in the set $\mathcal{C}_{\text{stab}}$. By definition the four transfer matrices in (\ref{eq:closed_loop}) are stable. Let $X=(I-GK)^{-1}$,  $Y=K(I-GK)^{-1}$, $W=(I-GK)^{-1}G$, $Z=(I-KG)^{-1}$.  Then, (\ref{eq:aff3}) holds by hypothesis. We now verify (\ref{eq:aff1})-(\ref{eq:aff2}).  
	 \begin{align*}
	 X-GY&=(I-GK)^{-1}-GK(I-GK)^{-1}=I\,,\\
	 W-GZ&=(I-GK)^{-1}G-G(I-KG)^{-1}=0\,,\\
	 -XG+W&=-(I-GK)^{-1}G+(I-GK)^{-1}G=0\,,\\
	 -YG+Z&=-K(I-GK)^{-1}G+(I-KG)^{-1}=I\,.
	 \end{align*}
	 Hence, (\ref{eq:aff1})-(\ref{eq:aff2}) are satisfied. Any $K \in \mathcal{C}_{\text{stab}}$ is thus represented by (\ref{eq:aff1})-(\ref{eq:aff3}). 
	 
2) Let $(X,Y,W,Z)$ satisfy (\ref{eq:aff1})-(\ref{eq:aff3}). Let $K=YX^{-1}$. Observe that $X$ is proper, but not strictly proper by (\ref{eq:aff1}). Hence, $X^{-1}$ is proper. Since $Y$ is proper by (\ref{eq:aff3}), it follows that $K=YX^{-1}$ is proper.  By (\ref{eq:aff1})-(\ref{eq:aff2}) we have
	 \begin{align*}
	 (I-GK)^{-1}&=(I-GYX^{-1})\\
	 &=(I-GY(I+GY)^{-1})^{-1} \\
	 &=(I+GY)=X\,,\\ 
	 K(I-GK)^{-1}&=YX^{-1}X=Y\,, \\
	 	 (I-GK)^{-1}G&=XG=W\,,\\
	 (I-KG)^{-1}&=I+KG(I-KG)^{-1} \quad~~ \\
	 &=I+YG=Z\,.
	 \end{align*}
	 Since $(X,Y,W,Z)$ are stable by (\ref{eq:aff3}), so are  the four transfer matrices  of (\ref{eq:closed_loop}).  We conclude that $K=YX^{-1}\in \mathcal{C}_{\text{stab}}$.

	 \subsection{Proof of Theorem~\ref{th:Youla_eq}}
1) Let $Q \in \mathcal{RH}_\infty^{m \times p}$ and consider the transfer matrices specified in (\ref{eq:Q_to_X_1})-(\ref{eq:Q_to_X_2}). Clearly, these transfer matrices are stable. We next verify (\ref{eq:aff1}). Notice that by (\ref{eq:dp2}) we have $U_r-GV_r=M_l^{-1}$. Hence,
\begin{align*}
X-GY&=(U_r-N_rQ-GV_r+GM_rQ)M_l\\
&=(M_l^{-1}-N_rQ+N_rM_r^{-1}M_rQ)M_l=I\,,
\end{align*}
and
\begin{align*}
W-GZ&=U_rN_l-N_rQN_l-G-GV_rN_l+GM_rQN_l\,,\\
&=M_l^{-1}N_l-G-N_rQN_l+N_rM_r^{-1}M_rQN_l=0\,.
\end{align*}
Finally, we verify (\ref{eq:aff2}). We have
\begin{align*}
-XG+W&=-U_rM_lG+N_rQM_lG+U_rN_l-N_rQN_l\\
&=(N_rQ-U_r)(M_lG-N_l)=0\,,
\end{align*}
and
\begin{align*}
-YG+Z&=\hspace{-0.05cm}-(V_r\hspace{-0.05cm}-\hspace{-0.05cm}M_rQ)M_lG+I+(V_r\hspace{-0.05cm}-\hspace{-0.05cm}M_rQ)N_l=I\,.
\end{align*}
Last, it is immediate to verify that $YX^{-1}=(V_r-M_rQ)(U_r-N_rQ)^{-1}$.

2) Let $(X,Y,W,Z)$ satisfy (\ref{eq:aff1})-(\ref{eq:aff3}) and $Q$ be chosen as per (\ref{eq:Youla_with_XYWZ}).  Clearly, $Q \in \mathcal{RH}_\infty^{m \times p}$. It remains to verify that  $(V_r-M_rQ)(U_r-N_rQ)^{-1}=YX^{-1}$. Notice that by (\ref{eq:aff1})-(\ref{eq:aff2}) we have $X=I+GY$, $W=G+GYG$ and $Z=I+YG$. Hence, 
\begin{align*}
Q=(U_l-V_lG)V_r-(U_l-V_lG)Y(U_r-GV_r)\,.
\end{align*}
Since we have $U_l-V_lG=M_r^{-1}$ by (\ref{eq:dp2}), we conclude
\begin{equation*}
Q=M_r^{-1}V_r-M_r^{-1}YM_l^{-1}\,,
\end{equation*}
from which we deduce $Y=(V_r-M_rQ)M_l$.
By using  the relationship $U_r-GV_r=M_l^{-1}$, it follows that
\begin{equation*}
X=I+GY=(U_r-N_rQ)M_l\,,
\end{equation*}
thus proving that $YX^{-1}=(V_r-M_rQ)(U_r-N_rQ)^{-1}$.
 


			\bibliographystyle{IEEEtran}

		\scriptsize{
	\bibliography{IEEEabrv,references}
	}

	\end{document}